\relax
\documentclass{llncs} 
\usepackage{pgfplots,pgfplotstable}
\usepgfplotslibrary{statistics}

\usepgfplotslibrary{external}
\usetikzlibrary{external}
\definecolor{green1}{RGB}{100,250,100}
\definecolor{green2}{RGB}{66,183,66}
\definecolor{green3}{RGB}{33,117,33}
\definecolor{green4}{RGB}{0,50,0}

\pgfplotsset{every axis/.append style={
    xlabel near ticks,
    ylabel style={yshift=-0.2cm},
    label style={font=\small},
    tick label style={font=\small}}}

\usepackage[compatibility=false]{caption} 

\usepackage[utf8]{inputenc}
\usepackage{booktabs}
\usepackage{xspace}
\usepackage{amssymb}
\usepackage{amsmath}
\usepackage{subcaption}	
\usepackage{siunitx}

\usepackage[htt]{hyphenat}

\DeclareMathOperator{\dist}{dist}

\newcommand*{\Pa}{\mathrm{P}}
\newcommand*{\C}{\mathrm{C}}
\newcommand*{\nwspace}{\hspace*{.1em}}

\usepackage{todonotes}
\makeatletter
\renewcommand{\todo}[2][]{\tikzexternaldisable\@todo[#1]{#2}\tikzexternalenable}
\makeatother
\renewcommand{\todo}[2][]{}

\usepackage{xcolor}
\providecommand{\ignore}[1]{}

\begin{document}


\title{Understanding the Effectiveness of Data Reduction in Public Transportation Networks}

\author{Thomas Bl{\"a}sius \and Philipp Fischbeck \and Tobias Friedrich \and Martin Schirneck}
\authorrunning{T. Bl{\"a}sius et al.}
\institute{Hasso Plattner Institute, University of Potsdam, Potsdam, Germany\\ \email{firstname.lastname@hpi.de}}

\maketitle

\begin{abstract}
  Given a public transportation network of stations and connections,
  we want to find a minimum subset of stations such that each
  connection runs through a selected station.  Although this problem
  is NP-hard in general, real-world instances are regularly solved
  almost completely by a set of simple reduction rules.  To explain
  this behavior, we view transportation networks as hitting set
  instances and identify two characteristic properties, locality and
  heterogeneity.  We then devise a randomized model to generate
  hitting set instances with adjustable properties.  While the
  heterogeneity does influence the effectiveness of the reduction
  rules, the generated instances show that locality is the significant
  factor.  Beyond that, we prove that the effectiveness of the
  reduction rules is independent of the underlying graph structure.
  Finally, we show that high locality is also prevalent in instances
  from other domains, facilitating a fast computation of minimum
  hitting sets.
\end{abstract}

\section{Introduction}
\label{sec:introduction}

A public transportation network is a collection of stations along with
a set of connections running through these stations.  But beyond its
literal definition, via bus stops and train lines, it also carries
some of the geographical, social, and economical structure of the
community it serves. Given such a network, we want to select as few
stations as possible to \emph{cover} all connections, i.e., each
connection shall contain a selected station.  This and similar
covering problems arise from practical needs, e.g., when choosing
stations for car maintenance, but their solutions also reveal some of
the underlying structure of the network. Despite the fact that
minimizing the number of selected stations is NP-hard, there is a
surprisingly easy way to achieve just that on real-world instances:
Weihe~\cite{weihe1998covering} showed for the German railroad network
that two straightforward reduction rules simplify the network to a
very small core which can then be solved by brute force.  This is not
a mere coincidence. Experiments have shown the same behavior on
several other real-world transportation networks. Subsequently, the
reduction rules became the standard preprocessing routine for many
different covering problems. See the work of Niedermeier and
Rossmanith~\cite{niedermeier2003efficient},
Abu-Khzam~\cite{AbuKhzam10dHittingSet}, or Davies and
Bacchus~\cite{Davies11MAXSAT}, to name just a few.  This raises the
question as to why these rules are so effective.
Answering this question would
not only close the gap between theory and practice for the
specific problem at hand, but also has the potential to lead to new insights into
the networks' structure and ultimately pave the way for algorithmic advances 
in bordering areas.

Our methodology for approaching this question is as follows.  We first
identify two characteristic properties of real-world transportation
networks: \emph{heterogeneity} and \emph{locality}; see
Section~\ref{sec:hitt-set-persp} for more details.  Then we propose a
model that generates random instances resembling real-world instances
with respect to heterogeneity and locality.  We validate our model by
showing empirically that it provides a good predictor for the
effectiveness of the reduction rules on real-world instances.
Finally, we draw conclusions on why the reduction rules are so
effective by running experiments on generated instances of varying
heterogeneity and locality.  Moreover, we show that our results extend
beyond transportation networks to related problems in other domains.

For our model, we regard transportation networks as instances of the
\emph{hitting set problem}.  From this perspective, connections are
mere subsets of the universe of stations and we need to select one
station from each set.  Note that this disregards some of the
structure inherent to transportation networks: A connection is not
just a set of stops but a sequence visiting the stops in a particular
order.  In fact, the sequences formed by the connections are paths in
an underlying graph, which itself has rich structural properties
inherited from the geography.  Focusing on these structural
properties, we also consider the \emph{graph-theoretic perspective}.
The working hypothesis for this perspective is that the underlying
graphs of real-world transportation networks have beneficial
properties that render the instances tractable.  We disprove this
hypothesis by showing that the underlying graph is almost irrelevant.
This validates the hitting set perspective, which disregards the
underlying graph.


In Section~\ref{sec:results}, we formally state our findings on the
graph-theoretic as well as the hitting set perspective.  We study the
hitting set instances of European transportation networks in
Section~\ref{sec:analysis-real-world}, identifying heterogeneity and
locality as characteristic features.  In
Section~\ref{sec:analysis-generated}, we define and evaluate a model
generating instances with these features.
Section~\ref{sec:other_domains} extends our findings to other domains
and Section~\ref{sec:conclusion} concludes this work.

\section{Preliminary Considerations}
\label{sec:results}

Before discussing the results regarding the two different
perspectives, we fix some notations and state the reduction rules
introduced by Weihe~\cite{weihe1998covering}. A \emph{public
transportation network} (or simply a \emph{network}) $N \,{=}\, (S, C)$
consists of a set $S$ of \emph{stations}
and a set $C$ of \emph{connections} which are sequences of stations.
That is, each connection $c\in C$ is a subset of $S$ together with a linear
ordering of its elements. Two stations $s_1, s_2 \in S$ are
\emph{connected} in $N$ if there exists a sequence of stations starting with
$s_1$ and ending in $s_2$ such that each pair of consecutive
stations shares a connection. The subnetworks induced by this equivalence relation
are called the \emph{connected components} of $N$.
Given $N \,{=}\, (S, C)$, the \textsc{Station Cover} problem is to find a
subset $S' \subseteq S$ of minimum cardinality such that each connection
is \textit{covered}, i.e., $S' \cap c \not= \emptyset$ for every $c \in C$.
%
%
The reduction rules by Weihe~\cite{weihe1998covering} are
based on notions of dominance, both between stations and connections.
For two different stations $s_1, s_2 \in S$,
$s_1$ \emph{dominates} $s_2$ if every connection containing $s_2$ also
contains $s_1$. If so, there is always an optimal station cover without
$s_2$, so it is never worse to select $s_1$ instead.  Thus, removing
$s_2$ from $S$ and from every connection in $C$ yields an
equivalent instance.  Similarly, for two different connections
$c_1, c_2 \in C$, $c_1$ \emph{dominates} $c_2$ if $c_1 \subseteq c_2$.
Every subset of $S$ covering $c_1$ then also covers $c_2$.  Removing
$c_2$ does not destroy any optimal solutions.
Weihe's algorithm can thus be summarized as follows.  Iteratively
remove dominated stations and connections until this is no longer
possible.  The remaining instance, the \emph{core}\footnote{We note
  that the core is unique up to automorphisms.  In particular, its
  size is independent of the removal order.}, is solved using brute
force.  Each connected component can be solved independently and the
running time is exponential only in the number of stations.  Thus, the
\emph{complexity} of an instance denotes the maximum number of
stations in any of its connected components.

The proofs of this section are in Appendix~\ref{sec:appendix-treewidth}.

\subsection{Graph-Theoretic Perspective}
\label{sec:graph-theor-persp}

One way to represent a network $N = (S, C)$ is via an undirected graph $G_N$
defined as follows. The stations $S$ are the vertices of $G_N$; for each
connection $(s_1, \dots, s_k) \in C$, $G_N$ contains the edges
$\{s_i, s_{i+1}\}_{1 \le i < k}$. The basic hypothesis of the
graph-theoretic perspective is that certain properties of $G_N$ make
the real-world \textsc{Station Cover} instances easy.

Consider a leaf $u$ in $G_N$, i.e., a degree-1 vertex.  If there is a
connection that contains only $u$, then this dominates all other
connections containing $u$.  Otherwise, all connections that contain
$u$ also contain its unique neighbor.  Thus, $u$ is dominated and
removed by the reduction rules.  We obtain the following proposition.
The \emph{2-core} is the subgraph obtained by iteratively removing
leaves~\cite{seidman83minimumDeg}.

\begin{proposition}
  \label{prop:reduce-to-2-core}
  The reduction rules reduce any \textsc{Station Cover} instance $N$
  to an equivalent instance $N'$ such that $G_{N'}$ is a subgraph of
  the 2-core of $G_{N}$, with additional isolated vertices.
\end{proposition}


\noindent
Proposition~\ref{prop:reduce-to-2-core} identifies the number of
vertices in the \mbox{$2$-core} of $G_N$ as an upper bound for the
\emph{core complexity}.  The following theorem shows that this bound
is arbitrarily bad. Supporting this assessment, we will see in
Section~\ref{sec:analysis-real-world} that the $2$-cores of the graphs
of real-world instances are rather large, while their core complexity
is significantly smaller.

\begin{theorem}
  \label{thm:reduction-to-2-core}
  For every graph $G$, there exist two \textsc{Station Cover}
  instances $N_1$ and $N_2$ with $G = G_{N_1} = G_{N_2}$ such that the
  core of $N_1$ has complexity~$1$ while the core of $N_2$ corresponds
  to the $2$-core of $G$.
\end{theorem}

\noindent
Theorem~\ref{thm:reduction-to-2-core} disproves the working hypothesis
of the graph-theoretic perspective.  For any connected graph that has
no leaves, there is a \textsc{Station Cover} instance that is
completely solved by the reduction rules, and another instance on the
very same graph that is not reduced at all.  Furthermore, unless the
$2$-core is small, the theorem shows that it is impossible to tell
whether or not the reduction rules are effective on a given instance
by only looking at the graph.

So far, we have only focused on Weihe's algorithm.  While our main goal is
to explain the performance of this algorithm, one could argue that
other methods exploiting different graph-theoretic properties are
better suited to solve real-world instances. The next theorem,
however, indicates that this is not the case.  Even on ``tree-like''
graph classes \textsc{Station Cover} remains NP-hard.  The reduction
used to prove this theorem was originally given by
Jansen~\cite{jansen2015structural}.


\begin{theorem}[\cite{jansen2015structural}, Theorem~5]
  \label{thm:treewidth}
  \textsc{Station Cover} is NP-hard even if the corresponding graph
  has treewidth~3 or feedback vertex number~2.
\end{theorem}

\subsection{Hitting Set Perspective}
\label{sec:hitt-set-persp}

Another way to represent a network $N$ is by an instance of the
\textsc{Hitting Set} problem.  Here, the connections $C \subseteq 2^S$
are regarded only as sets of stations (ignoring their order). An
optimal cover is a minimum-cardinality subset of $S$ that has a
non-empty intersection with all members of $C$.  This perspective
turns out to be much more fruitful.  In the next section, we analyze
the \textsc{Hitting Set} instances stemming from $12$ real-world
networks.  To summarize our results, we observe that the instances are
heterogeneous, i.e., the number of connections containing a given
station varies heavily.  Moreover, the instances exhibit a certain
locality, which probably has its origin in the stations' geographic
positions.

In more detail, for a station $s \in S$, let the number of connections
in $C$ that contain $s$ be the \emph{degree} of $s$.  Conversely, for
$c \in C$, $|c|$ is its \emph{degree}.  The connection degrees of the
real-world instances are rather homogeneous, i.e., every connection
has roughly the same size.  Although there are different types of
connections they appear to have a similar number of stops.  The
station degrees, on the other hand, vary strongly.  In fact, we
observe that the station degree distributions roughly follow a power
law.  This is in line with observations that, e.g., the sizes of
cities are power-law distributed~\cite{gabaix1999zipf}.  To quantify
the locality of an instance, we use a variant of the so-called
\emph{bipartite clustering coefficient}~\cite{robins2004small}.

We conjecture that heterogeneity of stations and locality of the
network are the crucial factors that make the reduction so effective.
If the station degrees vary strongly, chances are that some
high-degree station exists that dominates many low-degree ones.
Moreover, if locality is high, there tend to be several connections
differing only in few stations and stations appearing in similar sets
of connections. This increases the likelihood of dominance among the
elements of both $S$ and $C$.  To verify this hypothesis empirically,
we propose a model for generating instances of varying heterogeneity
and locality.  Our findings suggest that higher heterogeneity
decreases the core complexity, but the deciding factor is the
locality.  Finally, we observe that locality is also prevalent in
other domains.  As predicted by our model, preprocessing also greatly
reduces these instances.

\section{Analysis of Real-World Networks}
\label{sec:analysis-real-world}



\begin{table}[t!]
\small
  \centering
  \renewcommand{\tabcolsep}{7pt}
  \begin{tabular}{lr
        S[table-format=2.1] 
        S[table-format=1.1] 
        S[table-format=1.1] 
        S[table-format=1.2] 
        S[table-format=1.2] 
        cc}
    Data Set           & $|S|$\ \     & $\tfrac{|S|}{|C|}$ & $\delta_S$ & $\beta$ & KS   & $\kappa$ & \hspace{-3.5pt}$2$-core\hspace{-3.5pt} & core\ \ \\
    \midrule
    \texttt{\scriptsize sncf}      & \SI{1742}{}  & 4.0             & 2.2        & 3.3     & 0.03 & 0.47     & 70\%& 0.3\%   \\
    \texttt{\scriptsize nl}        & \SI{4558}{}  & 13.2            & 1.5        & 3.8     & 0.04 & 0.40     & 70\%& 2.8\%   \\
    \texttt{\scriptsize kvv}       & \SI{2115}{}  & 8.0             & 2.1        & 3.5     & 0.03 & 0.48     & 72\%& 0.8\%   \\
    \texttt{\scriptsize vrs}       & \SI{5491}{}  & 10.7            & 1.9        & 3.5     & 0.03 & 0.27     & 83\%& 0.1\%   \\
    \texttt{\scriptsize rnv}       & \SI{705}{}   & 12.4            & 1.4        & 4.2     & 0.06 & 0.38     & 54\%& 0.1\%   \\
    \texttt{\scriptsize athens}    & \SI{5729}{}  & 24.4            & 1.8        & 3.9     & 0.04 & 0.30     & 89\%& 4.7\%   \\
    \texttt{\scriptsize petersburg}  & \SI{4264}{}  & 6.5             & 2.5        & 4.0     & 0.03 & 0.31     & 86\%& 8.3\%   \\
    \texttt{\scriptsize warsaw}    & \SI{3944}{}  & 13.0            & 1.8        & 5.9     & 0.05 & 0.29     & 80\%& 5.9\%   \\
    \texttt{\scriptsize luxembourg}   & \SI{2484}{}  & 7.3             & 2.7        & 2.9     & 0.02 & 0.25     & 84\%& 0.2\%   \\
    \texttt{\scriptsize switzerland} & \SI{22535}{} & 5.6             & 2.0        & 4.5     & 0.02 & 0.33     & 71\%& 1.7\%   \\
    \midrule                                                                   
    \texttt{\scriptsize vbb}       & \SI{3031}{}  & 16.5            & 1.4        & 12.4    & 0.05 & 0.38     & 73\%& 1.8\%   \\
    \texttt{\scriptsize db}        & \SI{514}{}   & 0.9             & 15.7       & 2.0     & 0.07 & 0.28     & 78\%& 0.2\%   \\
  \end{tabular}
  \caption{Transportation networks with atypical instances
    separated. Shown are the number $|S|$ of stations, the
    station-connection ratio $|S|/|C|$, average station degree
    $\delta_S$, estimated power-law exponent $\beta$, corresponding KS
    distance, bipartite clustering coefficient $\kappa$, relative
    $2$-core size, and the relative core complexity.}
  \label{table:real-data}
\end{table}

We examined several public transportation networks from different cities
(\texttt{athens}, \texttt{petersburg}, \texttt{warsaw}), rural areas
(\texttt{sncf}, \texttt{kvv}, \texttt{vrs}, \texttt{rnv},
\texttt{vbb}), and countries (\texttt{nl}, \texttt{luxembourg},
\texttt{switzerland}, \texttt{db}). The networks are taken from the
\url{transitfeeds.com} repository. The raw data has the General
Transit Feed Specification (GTFS) format.  It stores multiple
connections for the same route, one for each time a train actually
drives that route.  For each route, only one connection was used.
\todo{Table~\ref{table:real-data}}Table~\ref{table:real-data} gives an overview of the relevant features
of the resulting networks.

We reduced each instance to its largest component.  For most of them,
only a small fraction of stations and connections are disconnected
from this component.  A notable exception is the
\texttt{vbb}-instance, representing the public transportation network
of the city of Berlin, Germany.  In total, it has \SI{13424}{}
stations while its largest component has only \SI{3031}{}.  The reason
is that different modes of transport are separated in the raw data.
As a result, \texttt{vbb} has rather uncommon features.
Another unusual case is the \texttt{db}-instance of the German railway
network.  \todo{Table~\ref{table:real-data}}Table~\ref{table:real-data} shows that most instances have a
station-connection ratio $|S|/|C|$ of roughly $10$.  For \texttt{db},
however, this ratio is at $0.9$ much smaller.

\subsubsection{Heterogeneity.}

The average station degree $\delta_S$ of the investigated networks
is a small constant around $2$, independent of the instances'
complexity. The only exception is the \texttt{db}-network. This 
can be explained by the atypical value for $|S|/|C|$, and that each
station is contained in much more connections. The average
connection degree $\delta_C = \delta_S \cdot |S|/|C|$ (not
explicitly given in the table) is roughly $20$, due to the
station-connection ratios all being of the same order.

\begin{figure}[t!]
    \begin{minipage}[t]{0.49\linewidth}
      \includegraphics{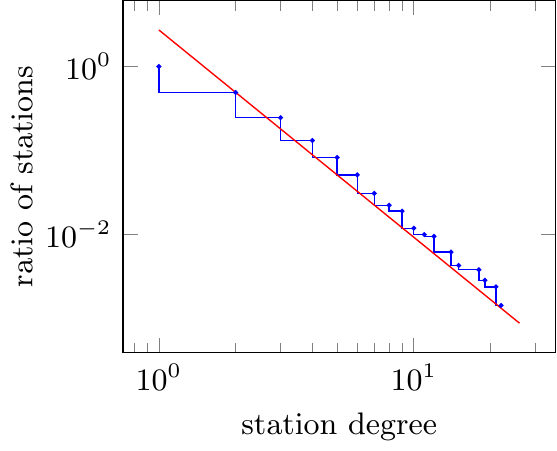}
    \end{minipage}%
    \hfill
  \begin{minipage}[t]{0.49\linewidth}
  	\includegraphics{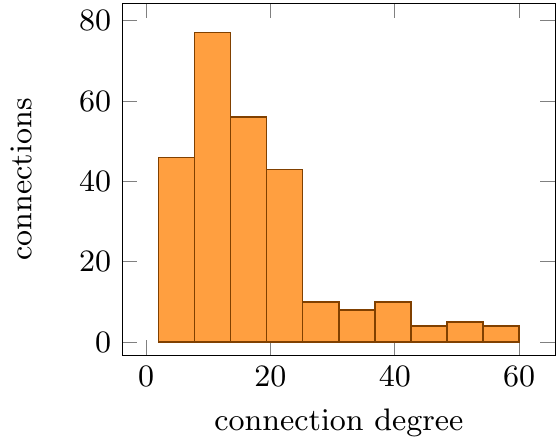}
  \end{minipage}
  \caption{\textbf{(left)}~The blue line is the CCDF of station
    degrees for the data set \texttt{kvv}.  The red line is the
    estimated power-law distribution. \textbf{(right)}~The histogram
    of connection degrees for the data set \texttt{kvv}.}
  \label{fig:histograms-kvv}
\end{figure}

Beyond small average degree, almost all instances exhibit strong
heterogeneity among the degrees of different stations. We take a
closer look at the \texttt{kvv}-instance as a prototypical example,
representing the public transportation network of Karlsruhe, Germany.
\todo{Figure~\mbox{\ref{fig:histograms-kvv}~(left)}}Figure~\mbox{\ref{fig:histograms-kvv}~(left)} 
shows the \emph{complementary cumulative distribution function} (CCDF)
of the station degrees in a log-log plot.  For a given value $x$, the
CCDF describes the share of stations that have degree at least~$x$.
The CCDF closely resembles a straight line (in log-log scaling),
indicating a \emph{power-law distribution}. That means, there exists a
real number $\beta$, the \textit{power-law exponent}, such that the
number of stations of degree $x$ is roughly proportional to
$x^{-\beta}$.  We estimated the power-law exponents using the
\texttt{python} package \texttt{powerlaw}~\cite{alstott2014powerlaw}.
For \texttt{kvv}, the exponent $\beta$ is approximately $3.5$.
%
%
The goodness of fit is measured by the \emph{Kolmogorov--Smirnov
  distance} (KS distance), which is the maximum absolute difference
between the CCDFs of the measurement and of the assumed
distribution. The KS distance for the \texttt{kvv} is $0.05$.
\todo{Table~\ref{table:real-data}}Table~\ref{table:real-data} reports both the power-law exponents and
the corresponding KS distances. The estimated values of $\beta$,
excluding the outlier \texttt{vbb}, indicate a high level of
heterogeneity. As a side note, the power-law exponent for \texttt{vbb}
is $4.1$ with a KS distance of $0.03$ when considering the whole
network instead of the largest component.
In contrast, the connection degrees are rather homogeneous,
cf.~e.g.~\texttt{kvv} in
\todo{Figure~\mbox{\ref{fig:histograms-kvv}~(right)}}Figure~\mbox{\ref{fig:histograms-kvv}~(right)}.  A possible
explanation is that long-distance trains stop less frequently.

\subsubsection{Locality.}

To measure locality, we adapt the \emph{bipartite clustering
  coefficient}~\cite{robins2004small}. Intuitively, it states how
likely it is that two stations which share a connection are also
contained together in a different connection, or that two connections
containing the same station also have another station in common.  For
a formal definition, first note that we can interpret a
\textsc{Hitting Set} instance $(S,C)$ as a bipartite graph with the
two partitions $S$ and $C$ and an edge joining $s \in S$ and $c \in C$
iff $s \in c$.  Let $\#_{\Pa_3}$ denote the number of paths of
length~$3$ and $\#_{\C_4}$ the number of cycles of length~$4$ in this
graph.  The bipartite clustering coefficient $\kappa$ then is defined
as $\kappa = 4 \,{\cdot}\, \#_{C_4}/\#_{\Pa_3}$.  It is the
probability that a uniformly chosen \mbox{$3$-path} is contained in a
\mbox{$4$-cycle}.  Before computing $\kappa$, we normalize the
bipartite graph by reducing it to its 2-core, which removes any attached
trees.  In doing so, the measure becomes more robust for our purpose, as
attached trees do not impact the difficulty of an instance (they get
removed by the reduction rules) while they decrease the clustering
coefficient.

%
%

The clustering coefficients are reported in
\todo{Table~\ref{table:real-data}}Table~\ref{table:real-data}.  All instances have a clustering
coefficient of at least $0.25$, which indicates a high level of
locality.  A possible explanation are the underlying geographic
positions of the stations, with nearby stations likely appearing in
the same connection.

\subsubsection{Degree of Reduction.}

We measure the effectiveness of the reduction rules using the
\emph{relative core complexity}.  It is the percentage of stations
that remain after exhaustively applying the preprocessing.
\todo{Table~\ref{table:real-data}}Table~\ref{table:real-data} shows that the resulting relative core
complexity is very low for all 12 instances.  This is in line with the
original findings of Weihe~\cite{weihe1998covering}, who applied the
reduction rules on a few select European train networks.  Moreover, it
generalizes these results to networks of different scales, from urban
to national.  On the other hand, the $2$-core is typically not much
smaller than the original instance.  This shows that
Proposition~\ref{prop:reduce-to-2-core} cannot explain the
effectiveness of the reduction rules, which supports our previous
assessment that the graph-theoretic perspective is not sufficient.


Judging from \todo{Table~\ref{table:real-data}}Table~\ref{table:real-data}, we believe that
heterogeneity of the stations and high locality are the crucial
properties rendering the preprocessing so effective.  Notwithstanding,
it is also worth noting that the reduction rules work well on all
instances, including \texttt{vbb} which is not very heterogeneous.
The clustering on the other hand is high for all instances, indicating
that locality is more important.
Also, the \texttt{db} and \texttt{vbb} outliers seem to show that the
influence of the station-connection ratio and the average station
degree is limited.  Though looking at these $12$ networks can provide
clues to what features are most important, it is not sufficient to draw a
clear picture.  In the following, we thoroughly test the effect of
different properties on the effectiveness of the reduction rules
by generating instances with varying properties.

\section{Analysis of Generated Instances}
\label{sec:analysis-generated}

This section discusses the generation and analysis of artificial
\textsc{Hitting Set} instances.  First, we present our model of
generation which is based on the \textit{geometric inhomogeneous
  random graphs}~\cite{bringmann17GIRGLinearTime}.  It
allows creating networks with varying degree of heterogeneity and
locality.  We then analyze these instances with respect to the degree
of reduction.


\subsection{The Generative Model}
\label{subsec:generation_model}

In the field of network science, it is generally accepted that vertex
degrees in realistic networks are heterogeneous~\cite{voitalov2018scale}.
A power-law distribution can be explained, inter alia, by the
preferential attachment mechanism~\cite{barabasi1999emergence}.
Beyond the generation of heterogeneous instances, different models
have been proposed to also account for locality.  The latter models
typically use some kind of underlying geometry.  One of the earliest
works in that direction is by Watts and Strogatz~\cite{ws-c-98}.  More
recently, and closer to our aim, Papadopoulos et
al.~\cite{pks-pvsgn-12} introduced the concept of popularity vs.\
similarity, making the creation of edges more likely, the more popular
and similar the connected vertices are.  They also observed that these
two dimensions are naturally covered by the hyperbolic geometry,
leading to hyperbolic random graphs~\cite{krioukov2010hyperbolic}.
Bringmann, Keusch, and Lengler~\cite{bringmann17GIRGLinearTime}
generalized this concept to \emph{geometric inhomogeneous random
  graphs} (GIRGs).  There, each vertex has a geometric position and a
weight.  Vertices are then connected by edges depending on their
weights and distances.  Despite a plethora of models for generating
graphs, we are not aware of models generating heterogeneous
\textsc{Hitting Set} instances.  The closest is arguably the work by
Giráldez-Cru and Levy~\cite{j-lrsi-17}, who generate \textsc{SAT}
instances using the popularity vs.\ similarity paradigm.

To generate \textsc{Hitting Set} instances with varying heterogeneity
and locality, we formulate a randomized model based on GIRGs.  Each
station and connection has a weight representing its importance.
Moreover, stations and connections are randomly placed in a geometric
space.  The distance between stations and connections then provides a
measure of similarity.  In the \textsc{Hitting Set} instance, some
station $s$ is a member of connection $c$ with a probability
proportional to the combined weights of $s$ and $c$ and inverse
proportional to the distance between the vertices $s$ and $c$.  To
make this more precise, let \mbox{$w_S \colon S \to \mathbb R$} and
\mbox{$w_C \colon C \to \mathbb R$} be two weight functions; we omit
the subscript when no ambiguity arises.  For $s \in S$ and $c\in C$,
let $\dist(s, c)$ denote the geometric distance between the
corresponding vertices.  Finally, fix two positive constants $a,T > 0$.
Then, station $s$ is contained in connection $c$ with probability
\begin{equation}
  \label{eq:model}
  P(s, c) =\min\left\lbrace 1,\ \left( a\cdot\frac{w(s)w(c)}{\dist(s, c)}\right)^{1/T} \right\rbrace.
\end{equation}

The parameter $a$ governs the expected degree.  The
\textit{temperature} $T$ controls the influence of the geometry.  For
$T \,{\to}\, 0$ the method converges to a step model, where $s$ is
contained in $c$ if and only if
$\dist(s,c) \,{\le}\, a \nwspace w(s)w(c)$.
Larger temperatures soften this threshold, allowing $s\in c$ for
larger distances, and $s \notin c$ for smaller distances, with a low
probability.  Thus, $T$ influences the locality of the instance.

The remaining degrees of freedom are the choice of the underlying
geometry and the weights.  For the geometry, we use the unit circle.
Positions for stations and connections are drawn uniformly at random
from $[0, 1]$ and the distance between $x, y \in [0, 1]$ is
$\min\{|x \,{-}\, y|, 1 \,{-}\, |x \,{-}\, y|\}$.  This is arguably
the simplest possible symmetric geometry.

To choose the weights properly, it is important to note that the
resulting degrees are expected to be proportional to the
weights~\cite{bringmann17GIRGLinearTime}.  Thus, we mimic the
real-world instances by choosing uniform weights for the connections
and power-law weights, with varying exponent $\beta$, for the
stations.  It is not hard to see that for $\beta \to \infty$, the
latter converge to uniform weights as well.  In summary, adjusting
$\beta$ controls the heterogeneity.

\subsection{Evaluation}
\label{subsec:eval}

We generate artificial networks and measure the dependence of their
relative core complexity on the heterogeneity and locality.  The size
of an instance has three components: the (original) complexity $|S|$,
the station-connection ratio $|S|/|C|$, and the average station degree
$\delta_S$.  Note that these values also determine the number $|C|$ of
connections and average connection degree $\delta_C$.  From the model,
we have the two parameters we are most interested in, the power-law
exponent $\beta$ and the temperature $T$.  We also consider the limit
case of uniform weights for all vertices; slightly abusing notation,
we denote this by $\beta = \infty$.

For the main part of the experiments, we used $|S| = \SI{2000}{}$,
$|S|/|C| = 10.0$, and $\delta_S = 2.0$, leaning on the respective
properties for the real-world instances.  We let $T$ vary between $0$
and $1$ in increments of $0.05$, and $\beta$ between $2$ and $5$ in
increments of $0.25$.  For each combination, we generated ten samples.
In the following, we first validate the data.  Then we examine the
influence of heterogeneity and locality.  Afterwards we test whether
our findings still hold true for different station-connection ratios
and station degrees.

\subsubsection{Data Validation.}

There are two aspects to the data validation.  First, the instances
should approximately exhibit the properties we explicitly put in,
i.e., the values of $|S|$, $|S|/|C|$, $\delta_S$, and the power-law
behavior.  Second, the implicit properties should also be as expected.
In our case, we have to verify that changing $T$ actually has the
desired effect on the bipartite clustering coefficient~$\kappa$.

Concerning the complexity $|S|$, note that a sampled instance per se
does not need to be connected.  If it is not, we again only use the
largest component. Thus, when generating an instance with
$\SI{2000}{}$ stations, the resulting complexity is actually a bit
smaller. There are typically many isolated stations due to the small
average station degree, this is particularly true for small power-law
exponents.  However, the complexity of the largest component never
dropped below $\SI{1000}{}$ and usually was between $\SI{1400}{}$ and
$\SI{1700}{}$ provided that \mbox{$\beta>2.5$}.  The transition to the
largest component mainly meant ignoring isolated stations. Thus, also
the station-connection ratio $|S|/|C|$ decreases slightly. For
\mbox{$\beta > 3$} it was typically around $8$ and always at least
$7$.  For smaller $\beta$, it is never below $5$.

Recall that the average station degree $\delta_S$ is controlled by
parameter $a$ in Equation~\ref{eq:model}. It is a constant in the
sense that it is independent of the considered station-connection
pair. However, it does depend on other parameters of the model.  As
there is no closed formula to determine $a$ from $\delta_S$, we
estimated it numerically. This estimation incurred some loss in
accuracy but yielded values of $\delta_S$ between $1.9$ and $2.1$,
very close to the desired $\delta_S = 2$.  Transitioning to the
largest component typically slightly increases $\delta_S$, as the
largest component is more likely to contain stations of higher degree.
Anyway, $\delta_S$ never went above~$2.7$.

As with the real-world instances, we estimated the exponent $\beta$ of
the generated instances.  For small values, the estimates matched the
specified values.  For larger exponents, the gap increases slightly,
e.g., an estimated $\beta=5.7$ for an instance with predefined
parameter~$5.0$.

\begin{figure*}[tb]
  \centering
  \begin{minipage}{0.45\textwidth}
    \includegraphics{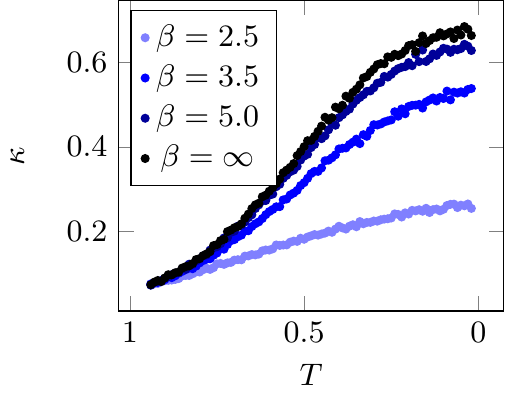}
  \end{minipage}\hfill
  \begin{minipage}{0.45\textwidth}
  	\includegraphics{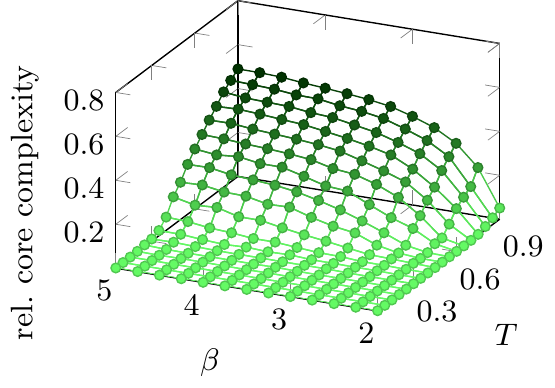}
  \end{minipage}
  \caption{\textbf{(left)}~The clustering coefficient $\kappa$
    depending on the temperature $T$ for different
    $\beta$. \textbf{(right)}~The relative core complexity depending
    on the power-law exponent $\beta$ and temperature $T$.}
  \label{fig:clustering_vs_temperature_and_3d}
\end{figure*}

\begin{figure*}[tb]
  \begin{minipage}[t]{0.49\textwidth}
  	\includegraphics{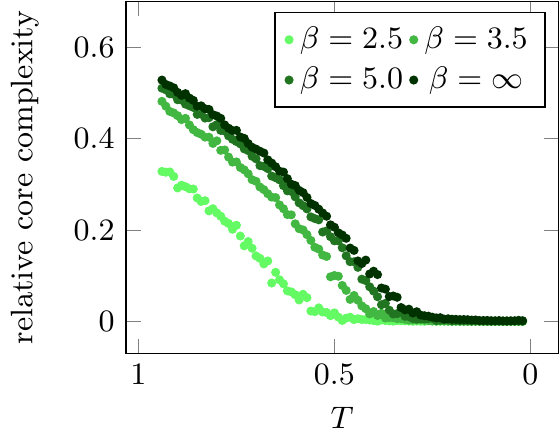}
  \end{minipage}\hfill
  \begin{minipage}[t]{0.49\textwidth}
  	\includegraphics{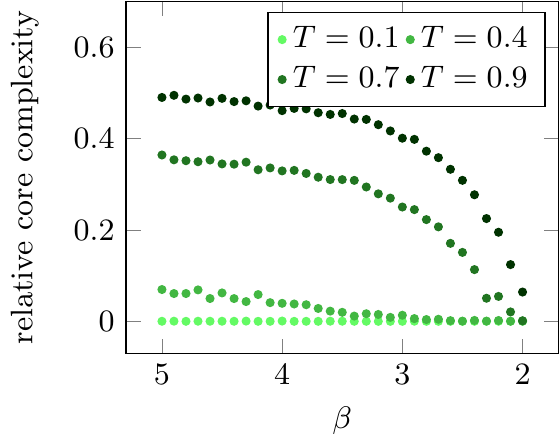}
  \end{minipage}
  \caption{The relative core complexity \textbf{(left)} depending on
    the temperature~$T$ (for different power-law exponents $\beta$),
    and \textbf{(right)}~depending on $\beta$ (for different $T$).}
  \label{fig:core-both-parameters}
\end{figure*}

Finally, we examined the dependency between the bipartite clustering
coefficient $\kappa$ and the temperature~$T$; see
\todo{Figure~\mbox{\ref{fig:clustering_vs_temperature_and_3d}~(left)}}Figure~\mbox{\ref{fig:clustering_vs_temperature_and_3d}~(left)}.  As
desired, the clustering coefficient increases with falling
temperatures.  More precisely, $\kappa$ ranges between $0$ and some
maximum attained at $T = 0$.  The value of this maximum depends on the
power-law exponent $\beta$, with smaller $\beta$ giving smaller
maxima.
  
\subsubsection{Heterogeneity and Locality.}

For each instance described above, we computed the relative core
complexity.  To reduce noise, we look at the arithmetic mean of ten
samples for each parameter configuration.  The measured core
complexities and clustering coefficients in fact showed only small
variance, with almost all values differing at most $5\%$ from the
respective mean.  These measurements are presented in
\todo{Figure~\mbox{\ref{fig:clustering_vs_temperature_and_3d}~(right)}}Figure~\mbox{\ref{fig:clustering_vs_temperature_and_3d}~(right)}, 
showing the mean core complexity depending on the temperature $T$ and
the power-law exponent $\beta$. For all parameter configurations, the
relative core complexity was at most $50\%$.  This is due to the low
average station degree, which leads to dominant low-degree stations.
More importantly, the core complexity varies strongly for different
values of $T$ and~$\beta$.


The complexity decreased both with lower temperatures and lower
power-law exponents. This further supports our claim that
heterogeneity and locality both have a positive impact on the
effectiveness of the reduction rules. The locality, however, seems to
be more vital.  To make this precise, low temperatures lead to small
cores, independent of $\beta$, as shown in
\todo{Figure~\mbox{\ref{fig:core-both-parameters}~(left)}}Figure~\mbox{\ref{fig:core-both-parameters}~(left)}. 
Even under uniform station weights ($\beta = \infty$), temperatures
below $0.3$ consistently produced instances with empty core.  In
contrast, the power-law exponent has only a minor impact.  One can see
in
\todo{Figure~\mbox{\ref{fig:core-both-parameters}~(right)}}Figure~\mbox{\ref{fig:core-both-parameters}~(right)} 
that for low temperatures, the core complexity is (almost) independent
of $\beta$.  For higher temperatures, the core complexities remain
high over wide ranges of $\beta$, except for very low exponents.

In summary, high locality seems to be the most prominent feature that makes
\textsc{Station Cover} instances tractable, independent of their heterogeneity.
Heterogeneity alone reduces the core complexity only slightly, except for extreme cases
(very low power-law exponents). It is thus not the crucial factor. 
In the following, we verify this general behavior
also for alternative model parameters such as station-connection ratio or average station degree.

\begin{figure*}[tb]
  \centering
  \begin{minipage}{0.45\textwidth}
  	\includegraphics{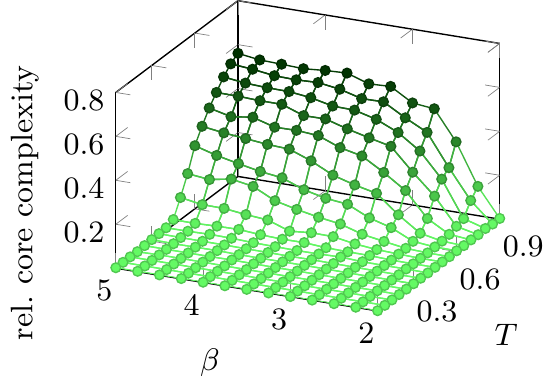}
  \end{minipage}\hfill
  \begin{minipage}{0.45\textwidth}
  	\includegraphics{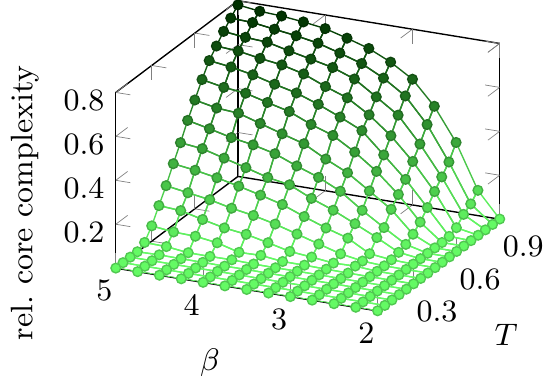}
  \end{minipage}
  \caption{The relative core complexity depending on the power-law
    exponent $\beta$ and temperature $T$ under alternative model
    parameters.  \textbf{(left)}~The station-connection ratio is
    $|S|/|C| = 4.0$ (instead of $10.0$). \textbf{(right)}~The average
    station degree is $\delta_S = 5.0$ (instead of $2.0$). }
  \label{fig:both-3d-different}
\end{figure*}

\subsubsection{Station-Connection Ratio.}

Recall that we fixed a ratio of $|S|/|C| = 10.0$ for the main part of our experiments.
To see
whether our observations are still valid for different settings, we
additionally generated data sets with mostly the same parameters as before,
except for $|S|/|C| = 4.0$.  The result are shown in
\todo{Figure~\mbox{\ref{fig:both-3d-different}~(left)}}Figure~\mbox{\ref{fig:both-3d-different}~(left)}. 
Comparing this to
\todo{Figure~\mbox{\ref{fig:clustering_vs_temperature_and_3d}~(right)}}Figure~\mbox{\ref{fig:clustering_vs_temperature_and_3d}~(right)}, 
one can see that general dependence on $T$
and $\beta$ is very similar. However, there are subtle differences. Under the smaller
station-connection ratio the instances are tractable even for larger
temperatures, up to $T = 0.5$ instead of the earlier $0.3$, i.e., for lower locality.
Also, the maximum core complexity over all combinations of $T$ and $\beta$ is larger than before,
reaching almost $60\%$, compared to the $50\%$ for $|S|/|C| = 10.0$.
In summary, in the (more realistic) low-temperature regime, a lower station-connection ratio seems to further improve the effectiveness of the reduction rules.

\begin{figure*}[tb]
  \centering
  \begin{minipage}[t]{0.45\textwidth}
  	\includegraphics{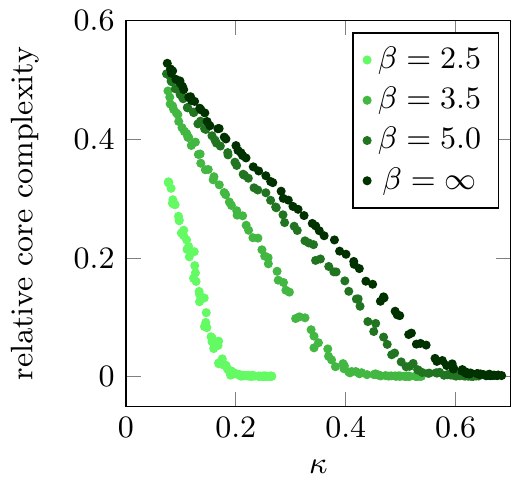}
  \end{minipage}\hfill
  \begin{minipage}[t]{0.45\textwidth}
  	\includegraphics{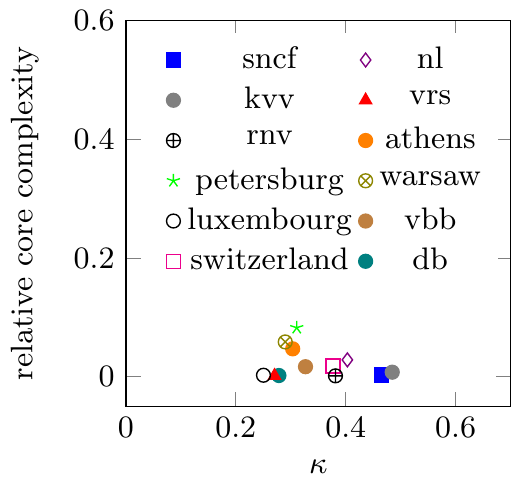}
  \end{minipage}
  \caption{The relative core complexity depending on the clustering
    coefficient $\kappa$ \textbf{(left)}~for generated instances with
    different values of the power-law exponent $\beta$, and
    \textbf{(right)}~for the real-world instances.}
  \label{fig:core-parameter}
\end{figure*}

\subsubsection{Average Degree.}

To examine the influence of the average station degree $\delta_S$, we
generated another instance set with the same parameters as the main
one, except that we increased the $\delta_S$ from $2.0$ to $5.0$.  The
results are shown in
\todo{Figure~\mbox{\ref{fig:both-3d-different}~(right)}}Figure~\mbox{\ref{fig:both-3d-different}~(right)}. 
Again the general behavior is similar, but now lower temperatures are
necessary to render the instances tractable.  Moreover, the maximum
core complexity increases significantly, reaching up to almost
$100\%$.  Generally, a smaller average degree makes the reduction
rules more effective.  Intuitively speaking, the existence of stations
with low degree increases the likelihood that the reduction rule of
station dominance can be applied.

\subsubsection{Comparison with Real-World Instances.}

To compare generated and real-world instances more directly, we
investigate the dependence between the relative core complexity and
the bipartite clustering coefficient $\kappa$, instead of the
model-specific temperature $T$.
\todo{Figure~\ref{fig:core-parameter}}Figure~\ref{fig:core-parameter} 
shows the results.  Several of the real-world networks are well
covered by the model. For example, our findings from the generated
instances can be directly transferred to \texttt{sncf} and
\texttt{kvv}.  They have power-law exponents of $3.3$ and $3.5$,
respectively, as well as clustering coefficients of at least $0.47$,
which explains their small core complexity.  Moreover, for
\texttt{luxembourg}, the model can explain the low complexity in spite
of a small clustering coefficient of $0.25$. The network exhibits a
small exponent $\beta = 2.9$, which benefits the effectiveness of the
preprocessing.  On the other hand, \texttt{petersburg} has clustering
$\kappa = 0.31$, but also a comparatively large core complexity of
over $8\%$. Here, the main factor seems to be the high power-law
exponent of $4.0$.

Notwithstanding, there are also some real-world instances that have an
unexpectedly low core complexity, which cannot be fully explained by
the model.  The \texttt{vrs}-instance has a low clustering coefficient
of $\kappa = 0.27$ and a high power-law exponent of $3.5$, but still a
very low $0.1\%$ core complexity.  The reason seems to be its low
average station degree of $\delta_S = 1.9$.  The
\texttt{switzerland}-instance also has a low core complexity of
$1.7\%$, despite its high power-law exponent of $\beta = 4.5$ and a
clustering coefficient of $\kappa = 0.33$.  Especially the high value
of $\beta$ would point to a much higher complexity, however, its
station-connection ratio $|S|/|C| = 5.6$ is significantly lower than
that of the generated instances.
%
All in all, the networks generated by the model are not perfectly realistic.
However, the model does replicate properties that are crucial for the effectiveness
of the reduction rules on real-world instances. Furthermore, in the interplay of heterogeneity and locality, it reveals locality as the more important property.

\section{Impact on Other Domains}
\label{sec:other_domains}

Although the focus of this paper is to understand which structural
properties of public transportation networks make Weihe's reduction
rules so effective, our findings go beyond that.  Our experiments on
the random model predict that \textsc{Hitting Set} instances in
general can be solved efficiently if they exhibit high locality.
Moreover, if the instance is highly heterogeneous, a smaller
clustering coefficient suffices; see
\todo{Figure~\mbox{\ref{fig:core-parameter}~(left)}}Figure~\mbox{\ref{fig:core-parameter}~(left)}.  The element-set ratio
and the average degree have, within reasonable bounds, only a minor
impact on the effectiveness.  Our experiments in
Section~\ref{subsec:eval} showed that instances are more difficult for
a larger average degree and if the difference between the number of
elements and the number of sets is high.  Experiments not reported in
this paper show that the latter is also true, if there are more sets
than elements (i.e., the ratio is below~$1$).

\subsubsection{Data Sets.}

We consider \textsc{Hitting Set} instances from three different
applications; see \todo{Table~\ref{table:real-data-other}}Table~\ref{table:real-data-other}.  The first set of
instances are metabolic reaction networks of Escherichia coli
bacteria.  The elements represent reactions and each set is a
so-called \emph{elementary mode}.  Analyzing the hitting sets of these
instances has applications in drug discovery.
%
%
The corresponding data sets, labeled \texttt{ec-*}, were generated
with the \texttt{Metatool}~\cite{kamp2006metatool}.
In the second type of instance, the sets consist of so-called
\emph{elementary pathways} that need to be hit by \emph{interventions}
that suppress all signals, which is relevant, inter alia, for the
treatment of cancer.  The data sets, \texttt{EGFR.*} and
\texttt{HER2.*}, were obtained via the \texttt{OCSANA}
tool~\cite{vera2013ocsana}.
%
The instances \texttt{country-cover} and \texttt{language-cover} are
based on a country-language graph, taken from the network collection
\texttt{KONECT}~\cite{kunegis2013konect}, with an edge between a
country and a language if the language is spoken in that country.  The
corresponding \textsc{Hitting Set} instances ask for a minimum number
of countries to visit to hear all languages, and for a minimum number
of languages necessary to communicate with someone in every country,
respectively.

\begin{table}[tb]
    \centering
    \renewcommand{\tabcolsep}{7pt}
    \small
    \begin{tabular}{lr
            r 
            S[table-format=1.2] 
            r
        }
        Data Set    		 & $|S|$ & $|S|/|C|$ & $\kappa$ & core \\
        \midrule
        \texttt{\scriptsize ec-acetate} 	& \SI{57}{}  & 0.214 & 0.67 & 1.8\% \\
        \texttt{\scriptsize ec-succinate} 	& \SI{57}{}  & 0.061 & 0.59 & 1.8\% \\
        \texttt{\scriptsize ec-glycerol} 	& \SI{60}{}  & 0.028 & 0.66 & 1.7\% \\
        \texttt{\scriptsize ec-glucose} 	& \SI{58}{}  & 0.009 & 0.66 & 36.2\% \\
        \texttt{\scriptsize ec-combined} 	& \SI{64}{}  & 0.002 & 0.62 & 40.6\% \\
        \texttt{\scriptsize EGFR.short} 	& \SI{50}{}  & 0.400 & 0.66 & 2.0\% \\
        \texttt{\scriptsize EGFR.sub} 		& \SI{56}{}  & 0.239 & 0.57 & 1.8\% \\
        \texttt{\scriptsize HER2.short} 	& \SI{123}{} & 0.230 & 0.53 & 9.8\% \\
        \texttt{\scriptsize HER2.sub} 		& \SI{172}{} & 0.068 & 0.55 & 10.4\% \\
		\texttt{\scriptsize country-cover} 	& \SI{248}{}  & 0.407 & 0.11 & 0.4\% \\
		\texttt{\scriptsize language-cover} & \SI{610}{} & 2.460 & 0.11 & 0.2\% \\
    \end{tabular}
    \caption{\textsc{Hitting Set} instances from other domains.
      Listed are the number $|S|$ of elements, the element-set ratio $|S|/|C|$, the bipartite clustering coefficient $\kappa$, and the relative core complexity.}
    \label{table:real-data-other}
\end{table}

\subsubsection{Evaluation.}

The basic properties of the instances and the effectiveness of the
reduction rules are reported in \todo{Table~\ref{table:real-data-other}}Table~\ref{table:real-data-other}.
The results match the prediction of our model: most instances have a
high clustering coefficient and the reduction rules are very
effective.  The only instances that stand out at first glance are
\texttt{ec-glucose}, \texttt{ec-combined}, \texttt{HER2.short}, and
\texttt{HER2.sub}, which are not solved completely by the reduction
rules despite their high clustering coefficients, as well as
\texttt{country-cover} and \texttt{language-cover}, which are solved
completely despite the comparatively low clustering coefficient of
$\kappa = 0.11$.

However, a more detailed consideration reveals that these instances
also match the predictions of the model.  First, the two instances
\texttt{country-cover} and \texttt{language-cover} are very
heterogeneous with power-law exponent $\beta = 2.2$.  As can be seen
in \todo{Figure~\ref{fig:clustering_vs_temperature_and_3d}}Figure~\ref{fig:clustering_vs_temperature_and_3d}~(left), a
clustering coefficient of $\kappa = 0.11$ is already rather high for
this exponent, leading to a low core complexity; see
\todo{Figure~\ref{fig:core-parameter}}Figure~\ref{fig:core-parameter}~(left).

The instances \texttt{ec-glucose} and \texttt{ec-combined} have skewed
element-set ratios (more than 100 times as many sets as elements) and
a high average degree ($\SI{30}{}$ for the sets; $\si{3k}$ and
$\si{13k}$ for the elements, respectively).  Thus, these instances at
least qualitatively match the predictions of the model that the
reduction rules are less effective if the element-set ratio is skewed
or the average degree is high.  One obtains a similar but less
pronounced picture for \texttt{HER2.short} and \texttt{HER2.sub}.


\section{Conclusion}
\label{sec:conclusion}

We explored the effectiveness of data reduction for \textsc{Station
  Cover} on transportation networks.  Our main finding is that
real-world instances have high locality and heterogeneity, and that
these properties make the reduction rules effective, with locality
being the crucial factor.  This directly transfers to general
\textsc{Hitting Set} instances.  For future work, it would be
interesting to rigorously prove that the reduction rules perform well
on the model.




{
\bibliographystyle{splncs04}
\bibliography{../references}

\begin{thebibliography}{10}
\providecommand{\url}[1]{\texttt{#1}}
\providecommand{\urlprefix}{URL }
\providecommand{\doi}[1]{https://doi.org/#1}

\bibitem{AbuKhzam10dHittingSet}
Abu-Khzam, F.N.: {A Kernelization Algorithm for $d$-Hitting Set}. Journal of
  Computer and System Sciences  \textbf{76},  524--531 (2010)

\bibitem{alstott2014powerlaw}
Alstott, J., Bullmore, E., Plenz, D.: {powerlaw: A Python Package for Analysis
  of Heavy-Tailed Distributions}. PLOS One  \textbf{9},  e85777 (2014)

\bibitem{barabasi1999emergence}
Barab{\'a}si, A.L., Albert, R.: {Emergence of Scaling in Random Networks}.
  Science  \textbf{286},  509--512 (1999)

\bibitem{bringmann17GIRGLinearTime}
Bringmann, K., Keusch, R., Lengler, J.: {Sampling Geometric Inhomogeneous
  Random Graphs in Linear Time}. In: Proceedings of the 25th Annual European
  Symposium on Algorithms (ESA). pp. 20:1--20:15 (2017)

\bibitem{Davies11MAXSAT}
Davies, J., Bacchus, F.: {Solving MAXSAT by Solving a Sequence of Simpler SAT
  Instances}. In: {Proceedings of the 17th International Conference on
  Principles and Practice of Constraint Programming (CP)}. pp. 225--239 (2011)

\bibitem{gabaix1999zipf}
Gabaix, X.: {Zipf's Law for Cities: an Explanation}. The Quarterly Journal of
  Economics  \textbf{114},  739--767 (1999)

\bibitem{j-lrsi-17}
Giráldez-Cru, J., Levy, J.: {Locality in Random SAT Instances}. In:
  Proceedings of the 26th International Joint Conference on Artificial
  Intelligence (IJCAI). pp. 638--644 (2017)

\bibitem{jansen2015structural}
Jansen, B.M.P.: {On Structural Parameterizations of Hitting Set: Hitting Paths
  in Graphs Using 2-SAT}. Journal of Graph Algorithms and Applications
  \textbf{21},  219--243 (2017)

\bibitem{kamp2006metatool}
von Kamp, A., Schuster, S.: {Metatool 5.0: Fast and Flexible Elementary Modes
  Analysis}. Bioinformatics  \textbf{22},  1930--1931 (2006)

\bibitem{krioukov2010hyperbolic}
Krioukov, D., Papadopoulos, F., Kitsak, M., Vahdat, A., Bogun{\'a}, M.:
  {Hyperbolic Geometry of Complex Networks}. Physical Review E  \textbf{82},
  036106 (2010)

\bibitem{kunegis2013konect}
Kunegis, J.: {KONECT - The Koblenz Network Collection}. In: Proceedings of the
  22nd International Conference on World Wide Web (WWW). pp. 1343--1350 (2013)

\bibitem{niedermeier2003efficient}
Niedermeier, R., Rossmanith, P.: {An Efficient Fixed-Parameter Algorithm for
  3-Hitting Set}. Journal of Discrete Algorithms  \textbf{1},  89--102 (2003)

\bibitem{pks-pvsgn-12}
Papadopoulos, F., Kitsak, M., Ángeles Serrano, M., Boguñá, M., Krioukov, D.:
  {Popularity Versus Similarity in Growing Networks}. Nature  \textbf{489},
  537--540 (2012)

\bibitem{robins2004small}
Robins, G., Alexander, M.: {Small Worlds Among Interlocking Directors: Network
  Structure and Distance in Bipartite Graphs}. Computational \& Mathematical
  Organization Theory  \textbf{10},  69--94 (2004)

\bibitem{seidman83minimumDeg}
Seidman, S.B.: {Network Structure and Minimum Degree}. Social Networks
  \textbf{5},  269--287 (1983)

\bibitem{vera2013ocsana}
Vera-Licona, P., Bonnet, E., Barillot, E., Zinovyev, A.: {OCSANA: Optimal
  Combinations of Interventions from Network Analysis}. Bioinformatics
  \textbf{29},  1571--1573 (2013)

\bibitem{voitalov2018scale}
Voitalov, I., van~der Hoorn, P., van~der Hofstad, R., Krioukov, D.V.:
  {Scale-free Networks Well Done}. CoRR  \textbf{abs/1811.02071} (2018)

\bibitem{ws-c-98}
Watts, D.J., Strogatz, S.H.: {Collective Dynamics of ‘Small-World’
  Networks}. Nature  \textbf{393},  440--442 (1998)

\bibitem{weihe1998covering}
Weihe, K.: {Covering Trains by Stations or the Power of Data Reduction}. In:
  Proceedings of the 1998 Algorithms and Experiments Conference (ALEX).
  pp.~1--8 (1998)

\end{thebibliography}
}

\clearpage
\appendix

\section{Omitted Proofs}
\label{sec:appendix-treewidth}

\setcounter{theorem}{0}

In this appendix, we sketch the omitted proofs of the theorems in Section~\ref{sec:results}.

\begin{theorem}
	For every graph $G$, there exist \textsc{Station Cover} instances
	$N_1$ and $N_2$ with $G = G_{N_1} = G_{N_2}$ such that the core of
	$N_1$ has complexity $1$ while the core of $N_2$ corresponds to the
	$2$-core of $G$.
\end{theorem}

\begin{proof}[Proofsketch]
	We assume that $G$ is connected, otherwise we apply the following to
	each component.  For $N_1$, we choose connections such that there is
	a single designated station $s$ contained in each connection.  Such
	connections can be constructed as follows.  For an edge $\{u, v\}$
	of $G$, construct a simple path that contains $s$ and the edge
	$\{u, v\}$.  We use the stations on this path as a connection in
	$N_1$ and repeat this for every edge.  It can be easily verified
	that $G_{N_1} = G$ and that the core of $N_1$ has complexity~1 ($s$
	dominates all other stations).
	
	For $N_2$, first assume that $G$ has no leaves, i.e., $G$ is already
	the $2$-core.  It is then not hard to see that having one connection
	for each edge containing exactly its two endpoints prohibits any of
	the two reduction rules.  In the general case, we have to be a bit
	more careful as applying reduction rules to leaves of $G$ can lead
	to connections that contain only a single station.  These then lead
	to further reductions that can potentially cascade through the
	$2$-core.  We prevent this as follows.  For every edge in the $2$-core
	of $G$, we add the connection containing its two endpoints (as
	before).  For every station $s$ not in the $2$-core, we find a simple
	path that contains $s$ and exactly one edge from the $2$-core.  We add
	the connection containing the stations of this path to $N_2$.  One
	can verify that indeed $G_{N_2} = G$ as every edge in $G$
	corresponds to a pair of consecutive stations in at least one
	connection.  Moreover, exhaustively applying the reduction rules
	leads to the $2$-core of $G$ with one connection for each edge.
\end{proof}

\begin{theorem}[\cite{jansen2015structural}, Theorem~5]
	\label{thm:treewidth-appendix}
	\textsc{Station Cover} is NP-hard even if the corresponding graph
	has treewidth~3 or feedback vertex number~2.
\end{theorem}

\begin{proof}[Proofsketch]  
	We prove this by reducing from \textsc{3-SAT}.  Let $\varphi$ be a
	propositional formula on $n$ variables.  We create two vertices $x$
	and $\overline{x}$ for every literal and connect them.  Then, we add
	two vertices $z_0$ and $z_1$ which are connected to every $x$
	vertex.  For every literal, we add a path $(x, \overline{x})$.  For
	every clause $(a \vee b \vee c)$ we add a path
	$(x_a, z_0, x_b, z_1, x_c)$ where each $x_i$ is the vertex
	$\overline{x_i}$ or $x_i$ based on whether the variable is inverted
	or not.  Figure~\ref{fig:treewidth-3} illustrates this reduction.  Now,
	if and only if there is a path cover by vertices of size~$n$ (where
	$n$ is the number of variables), the SAT instance is solvable.  The graph created has a feedback vertex number of 2: if we remove both $z_0$ and $z_1$ from the graph, the only edges left are those between the vertices $x$ and $\overline{x}$ for each variable. The graph also has treewidth 3. Consider the following tree decomposition. For every variable $x$, we add a bag $\{x, \overline{x}, z_0, z_1\}$. Our tree decomposition now consists of all of these bags in one single path in an
	arbitrary order. Each vertex is in some bag. For every edge, there is a bag that contains both of its vertices. Finally, the only vertices that occur in multiple bags are $z_0$ and $z_1$, and since they are in every bag, this forms a subtree. Each bag has a size of 4, thus the graph has a treewidth of 3.
	
	\begin{figure}[tb]
		\centering
		\includegraphics[width=0.7\textwidth]{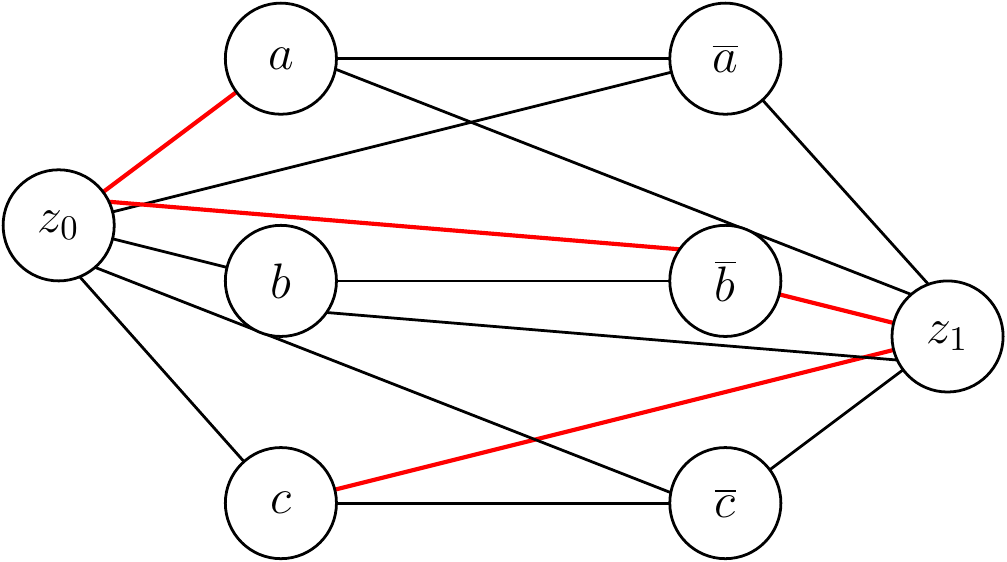}
		\caption{Illustration of the proof of Theorem~\ref{thm:treewidth-appendix}. The resulting graph has treewidth 3. The path shown in red represents the clause $a\vee \overline{b}\vee c$.}
		\label{fig:treewidth-3}
	\end{figure}
\end{proof}

\end{document}